\documentclass[10pt, twocolumn]{article}

\usepackage{authblk}
\usepackage[cmex10]{amsmath}
\usepackage{amsthm}
\usepackage{amssymb}
\usepackage{bbm}
\usepackage{graphicx}
\graphicspath{../}
\usepackage{hyperref}
 
\theoremstyle{plain}
\newtheorem{theorem}{Theorem}
\newtheorem{proposition}{Proposition}
\newtheorem{corollary}{Corollary}
\newtheorem{lemma}{Lemma}
 
\theoremstyle{definition}

\title{Information Spreading on Almost Torus Networks}
\author[1]{Antonia Maria Masucci
\thanks{Email: \href{mailto:antonia-maria.masucci@ensea.fr}{antonia-maria.masucci@ensea.fr}}}
\affil[1]{ETIS/ENSEA - Universit\'e de Cergy Pontoise - CNRS\\ 6, avenue du Ponceau\\ 95014 Cergy-Pontoise\\ France}
\author[2]{Alonso Silva
\thanks{Email: \href{mailto:alonso.silva@alcatel-lucent.com}{alonso.silva@alcatel-lucent.com}
To whom correspondence should be addressed.}}
\affil[2]{Alcatel-Lucent Bell Labs France\\ Centre de Villarceaux\\ Route de Villejust\\ 91620 Nozay\\ France}
\date{}

\begin{document}

\maketitle

\begin{abstract}

Epidemic modeling has been extensively used in the last years
in the field of telecommunications and computer networks.
We consider the popular Susceptible-Infected-Susceptible spreading model
as the metric for information spreading.
In this work, 
we analyze information spreading on a particular class of networks
denoted almost torus networks and over the lattice which can be
considered as the limit when the torus length goes to infinity.
Almost torus networks consist on the torus network topology
where some nodes or edges have been removed.
We find explicit expressions for the characteristic polynomial
of these graphs and tight lower bounds for its computation.
These expressions allow us to estimate their spectral radius 
and thus how the information spreads on these networks.

\end{abstract}

\section{Introduction}

There exists an abundant literature
on epidemic modeling and in particular on epidemics on networks (see e.g., the books 
\cite{BrauerDJ2008,DaleyG1999,DraiefM2010}
and references therein).
In the last decades,  epidemic modeling 
has been extensively used in the field of
telecommunications and computer networks.
For example, epidemic models have been applied in order
to analyze the spread of computer viruses 
and worms~\cite{KephartW1991},
they have also been applied on epidemic
routing for delay tolerant networks~\cite{ZhangNKT2007},~etc.
The structure of networks 
plays a critical role in the spread of a 
viral message and the authors of~\cite{BampoEMSW2008,GaneshMT2005,LeskovecAH2006}
have identified some conditions for successful performance of viral marketing.
In particular, the authors of~\cite{GaneshMT2005}
have shown that the spectral radius
of the graph determines the epidemic
lifetime and its coverage.
More recently on~\cite{Wei2012},
the authors have studied  the intertwined propagation
of two competing ``memes'' (or viruses, rumors, etc.)
in a composite network (individual agents
are represented across two planes e.g. Facebook and Twitter).
The authors have realized that the meme persistence $\delta$ and the meme strength $\beta$
on each plane, and the spectral radius of the graph,
completely determine which of the two competing ``memes'' prevail.

The topologies of focus in this work, torus network topologies, are commonly used in
the production of high-end computing systems~\cite{TangLDBY2011}. 
A number of supercomputers on the TOP500 list
use three-dimensional torus networks~\cite{TOP500}.
For instance,
IBM's Blue Gene L~\cite{BlueGeneL,Adiga2005} and Blue Gene P~\cite{BlueGeneP};
Cray's XT and XT3~\cite{Cray} systems use
three-dimensional torus networks for node communication. 
IBM Blue Gene/Q uses a five-dimensional torus network~\cite{Chen2011}.
Fujitsu's K computer and the PRIMEHPC FX10 use a 
proprietary six-dimensional torus interconnect called Tofu~\cite{Ajima2009}.
Torus networks are used because of a combination of their
linear per node cost scaling and their competitive overall
performance.

In this work, we analyze information spreading on a particular class of networks
denoted almost torus networks where we assume the popular
Susceptible-Infected-Susceptible model as 
model of information spreading.
Almost torus networks consist on the torus network topology
where some nodes or edges have been removed.
This situation can model the failure of some
computer nodes or connections between the computer nodes.
As we will see, in those graphs, the study of the
spectral radius of the graph is determinant in order to analyze the information spreading.
We find explicit expressions for the characteristic polynomial of these graphs and very tight
lower bounds for its computation. These expressions allow us to estimate
their spectral radius and thus how the information spreads on these networks.

The outline of the work is as follows.
In Section~\ref{sec:preliminaries}, we recall some preliminary notions of graph theory.
In Section~\ref{sec:modelSIS}, we present the Susceptible-Infected-Susceptible model of information spreading.
Then, in Section~\ref{sec:OneRemoved}, we analyze the information spreading 
in the almost torus network where one node has been removed and in Section \ref{sec:SetRemoved} we extend 
our results to the cases when a set of nodes has been removed 
and when an edge has been removed from the torus network.
In Section~\ref{sec:lower}, we provide lower bounds for the two-dimensional torus network.
In Section~\ref{sec:numerical}, we present numerical results that validate our analysis.
Finally, in Section~\ref{sec:conclusions} we conclude.

\section{Preliminaries}\label{sec:preliminaries}

Let $\mathcal{G}=(\mathcal{V},\mathcal{E})$ denote an undirected graph
with no self-loops.
We denote by \mbox{$\mathcal{V}=\mathcal{V}(\mathcal{G})=\{v_1,\ldots,v_n\}$}
the set of nodes and by \mbox{$\mathcal{E}=\mathcal{E}(\mathcal{G})\subseteq\mathcal{V}\times\mathcal{V}$}
the set of undirected edges of~$\mathcal{G}$.
If~$\{v_i,v_j\}\in\mathcal{E}(\mathcal{G})$ we call nodes $v_i$ and $v_j$
{\it adjacent} (or neighbors), which we denote by $v_i\sim v_j$.
We define the set of neighbors of node~$v$ as
\mbox{$\mathcal{N}_v=\{w\in\mathcal{V}: \{v,w\}\in\mathcal{E}\}$}.
The {\it degree} of a node~$v$, denoted by~$\mathrm{deg}_v$, corresponds to the number of neighbors of~$v$,
i.e. the cardinality of the set $\mathcal{N}_v$.
We define a {\it walk} of length~$k$ from~$v_0$ to~$v_k$ to be
an ordered sequence of nodes $(v_0,v_1,\ldots,v_k)$ such that
$v_i\sim v_{i+1}$ for $i=0,1,\ldots,k-1$.
If~$v_0=v_k$, then the walk is closed.

Graphs can be algebraically represented via matrices.
The {\it adjacency matrix} of an undirected graph~$\mathcal{G}$,
denoted by \mbox{$A=A(\mathcal{G})$},
is an $n\times n$ symmetric matrix defined entry-wise as
\begin{equation*}
A_{ij}=
\left\{
\begin{array}{rl}
1 & \quad\textrm{if $v_i$ and $v_j$ are adjacent,}\\
0 & \quad\textrm{otherwise.}
\end{array}
\right.
\end{equation*}

We recall the well-known result that for $k\in\mathbb{N}$,
$A^k_{ij}$
is the number of paths of length $k$ connecting the $i$-th
and $j$-th vertices (proof by induction).
Since $A^0$ is the identity matrix, we thus
accept the existence of walks of length zero.
We use $I$ to denote the identity matrix,
where its order is determined by the context.

We define the {\it Laplacian matrix}~$L$ for graphs without loops or multiple edges, as follows:
\begin{equation*}
L_{ij}=\left\{
\begin{array}{rl}
\mathrm{deg}_{v_i} & \quad\text{if $v_i=v_j$},\\
-1 & \quad\text{if $v_i$ and $v_j$ are adjacent},\\
0 & \quad\text{otherwise}.
\end{array}
\right.
\end{equation*}
We notice that the Laplacian of a graph can be written as
$L=D-A$ where 
$D$ is a diagonal matrix whose diagonal entries correspond to the degree of each node
and $A$ is the adjacency matrix.

The {\it spectral radius} of a graph $\mathcal{G}$, denoted $\rho(A)$,
is the size of the largest eigenvalue (in absolute value) of the adjacency matrix of the graph,
i.e. $\rho(A)=\max_i(\lvert\lambda_i\rvert)$.
Since $A$ is a symmetric matrix with nonnegative entries, all its eigenvalues are real.
The {\it characteristic polynomial} of $\mathcal{G}$,
denoted $\phi(\mathcal{G},x)$ is defined as ${\mathrm{det}}(xI-A)$,
that corresponds to the characteristic polynomial of the adjacency matrix~$A$.
The {\it walk generating function}\footnote{This can be viewed indifferently as a matrix
with rational functions as entries,
or as a formal power series in~$x$
over the  ring of all polynomials
in the matrix~$A$.}
$W(\mathcal{G},x)$
is defined to be \mbox{$(I-xA)^{-1}$}.
The $ij$-entry of $W(\mathcal{G},x)$ will be written
as~$W_{ij}(\mathcal{G},x)$.
If $\mathcal{S}$ is a subset of $\mathcal{V}(\mathcal{G})$
then $\mathcal{G}\setminus\mathcal{S}$ is the subgraph of $\mathcal{G}$
induced by the vertices not in $\mathcal{S}$.
We normally write
$\mathcal{G}\setminus i$ instead of $\mathcal{G}\setminus\{i\}$
and $\mathcal{G}\setminus ij$ instead of $\mathcal{G}\setminus\{i,j\}$.

\section{Model of Information Spreading}\label{sec:modelSIS}

We use the popular Susceptible-Infected-Susceptible (SIS) model of viral spreading~\cite{Hethcote2000}
as the metric for information spreading.
We remark that our results can be easily extended to the 
Susceptible-Infected$_1$-Infected$_2$-Susceptible (SI$_1$I$_2$S) model~\cite{Wei2012}
of spreading over composite networks.
We consider that each node can be in two possible states: susceptible (of being infected) or infected.
We denote these two states as~$\mathcal{S}$ and $\mathcal{I}$, respectively.

The result presented on this section was first found in~\cite{Wang2003} and~\cite{Chakrabarti2008},
through mean-field approximations of the Markov chain
evolution of the $2^n$ possible states.
We believe this alternative proof to be simpler and we presented here for completeness.

Consider a population of $n$ nodes interconnected via an undirected graph
$\mathcal{G}=(\mathcal{V},\mathcal{E})$. Time is slotted.
In each time slot, infected nodes attempt to contaminate
their susceptible neighbors, where each infection attempt
is successful with a probability~$\beta$, independent
of other infection attempts.
The parameter $\beta$ is called the virus birth rate (or meme strength as denoted in~\cite{Wei2012}).
Each infected node recover in time slot $t$ with probability~$\delta$.
The parameter $\delta$ is called virus curing rate (or meme persistence).
We notice that this parameter captures the meme persistence
in an inverse way, i.e. a high $\delta$ means low persistence.
This means a very contagious virus will be modeled with a low $\delta$ value.

We define $p_{i,t}$ as the probability that node~$i$
is infected at time~$t$.
We define $\zeta_{i,t+1}$ as
the probability that node~$i$
will not receive infections from
its neighbors in the next time-slot,
which is given by
\begin{equation}\label{eq:orsongo1}
\zeta_{i,t+1}=\prod_{j\in\mathcal{N}_i}(1-\beta p_{j,t}),
\end{equation}
where $\mathcal{N}_i$ denotes the set of neighbors of node~$i$.
This expression can be interpreted as the probability
that none of the infection attempts is successful.

The probability for a node~$i$ 
of not being infected at time~\mbox{$t+1$}
depends on the conditioning event
if the node~$i$ was or was not
infected at time~$t$:
\begin{enumerate}
\item\label{tumbuktu3} 
For the first case, if the node~$i$ was infected at time $t$,
then the probability for not being infected at time~$t+1$
is equal to $\delta\zeta_{i,t+1}$,

\item\label{tumbuktu2} 
For the second case, if the node~$i$ was not infected at time $t$,
then the probability for not being infected at time~$t+1$
is equal to $\zeta_{i,t+1}$.
\end{enumerate}

From \ref{tumbuktu3}) and \ref{tumbuktu2}),
the probability for a node~$i$ 
of not being infected at time $t+1$ 
is equal to
\begin{align*}
1-p_{i,t+1}
={}&\zeta_{i,t+1}(1-p_{i,t})+\delta\zeta_{i,t+1}p_{i,t}.
\end{align*}

Replacing from eq.~\eqref{eq:orsongo1}, we obtain that
\begin{equation}\label{eq:mali1}
1-p_{i,t+1}=(1-p_{i,t})\prod_{j\in\mathcal{N}_i}(1-\beta p_{j,t})+\delta p_{i,t}\prod_{j\in\mathcal{N}_i}(1-\beta p_{j,t}).
\end{equation}

We focus on the criterion based on the asymptotic stability
of the disease-free equilibrium $p_i^*(t)=0$ for all~$i$.
For doing this, we will make use of the following theorem.

\begin{figure}[h!]
 \centering
  \includegraphics[width=0.49\textwidth]{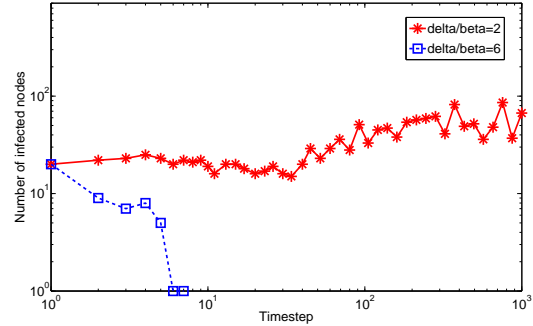}
  \caption{Number of infected nodes vs time.}
  \label{fig:number_of_infected_nodes_vs_time}
\end{figure}

\hspace{-32mm}
\begin{figure*}[ht]
  \begin{minipage}[b]{0.49\linewidth}
    \centering
    \includegraphics[angle=280,width=\linewidth]{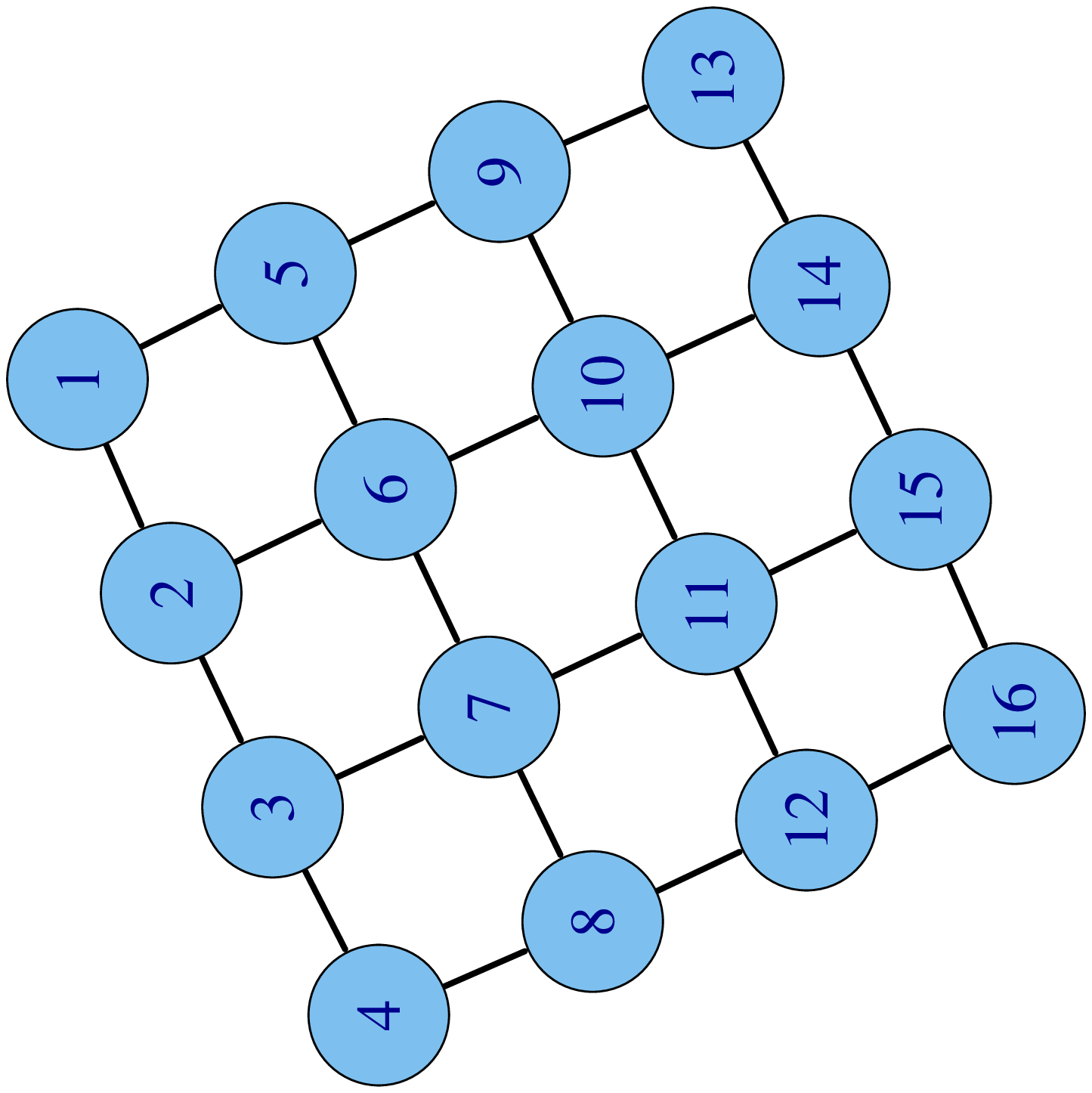}
    \caption{A two-dimensional grid of length $4$.}
    \label{fig:two_dimensional_grid}
  \end{minipage}
  \hspace{0.5cm}
  \begin{minipage}[b]{0.49\linewidth}
    \centering
    \includegraphics[angle=280,width=\linewidth]{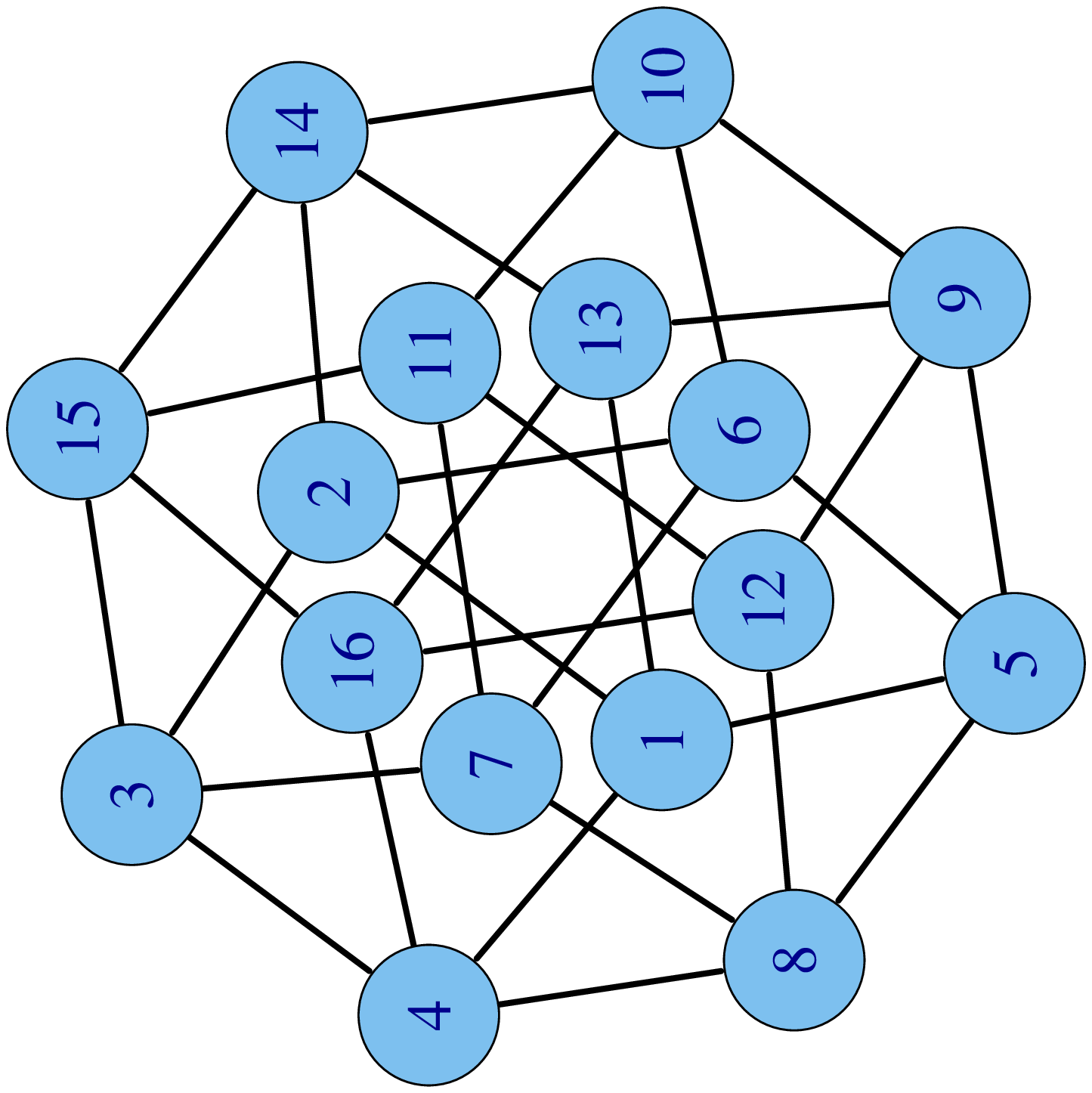}
    \caption{A two-dimensional torus of length $4$.}
    \label{fig:two_dimensional_torus}
  \end{minipage}
\end{figure*}

\begin{theorem}[Hirsch and Smale, 1974~\cite{Hirsch1974}]\label{thm:hirsch}
The system given by $\mathbf{p}_{t+1}=g(\mathbf{p}_t)$
is asymptotically stable at an equilibrium point~$\mathbf{p}^*$,
if the eigenvalues of the Jacobian $J=\nabla g(\mathbf{p})$
are less than $1$ in absolute value, where
\begin{equation*}
J_{k,l}=[\nabla g(\mathbf{p}^*)]_{k,l}
=\frac{\partial p_{k,t+1}}{\partial p_{l,t}}\bigg|_{\mathbf{p}_t=\mathbf{p}^*}.
\end{equation*}
\end{theorem}

We rewrite eq.~\eqref{eq:mali1} as follows:
\begin{equation*}
p_{i,t+1}=1-(1-p_{i,t})\prod_{j\in\mathcal{N}_i}(1-\beta p_{j,t})-\delta p_{i,t}\prod_{j\in\mathcal{N}_i}(1-\beta p_{j,t}).
\end{equation*}

Then
\begin{equation}\label{eq:mali2}
J_{k,l}=
\left\{
\begin{array}{rl}
1-\delta & \textrm{if}\quad k=l,\\
\beta & \textrm{if $l$ is neighbor of $k$},\\
0 & \textrm{otherwise},
\end{array}
\right.
\end{equation}
where we recall that the evaluating point is $\mathbf{p}^*=0$.

Eq.~\eqref{eq:mali2} can be written in a more compact way
as \mbox{$J=(1-\delta)I+\beta A$}.
Using Theorem~\ref{thm:hirsch}, we obtain that
for an asymptotic stability of the disease-free equilibrium
we need to impose that the eigenvalues of $(1-\delta)I+\beta A$
are in absolute value smaller than $1$,
or equivalently,
\begin{equation}\label{eq:spectral_timbuktu}
\rho(A)<\frac{\delta}{\beta}.
\end{equation}

In Fig.~\ref{fig:number_of_infected_nodes_vs_time},
we present two scenarios where
we consider a total population of $900$ nodes
and an initial seeding of $20$ nodes
with the same probability of infection \mbox{$\beta_1=\beta_2=0.1$}
but with different probabilities of recovery,
$\delta_1=0.2$ for the solid curve and
$\delta_2=0.6$ for the dashed curve.
We consider a graph with spectral radius $\rho(A)=4$.
In the second scenario, as predicted by the analysis and
equation~\eqref{eq:spectral_timbuktu}, the virus or information spreading
dies out, however for the first scenario the virus or information spreading
may continue as time increases.

\section{Information Spreading over the Torus with one removed node}\label{sec:OneRemoved} 

In a regular grid topology, each node in the network
is connected with at most two neighbors along one or more directions (see e.g. Fig.~\ref{fig:two_dimensional_grid}).
If the network is one-dimensional and we connect the first and last nodes, then the resulting
topology consists on a chain of nodes connected in a circular loop,
which is known as ring.
In general, when an $n$-dimensional grid 
network is connected circularly in more than one dimension, 
the resulting network topology is a torus (see e.g. Fig.~\ref{fig:two_dimensional_torus}).
In this work, we consider torus networks with the same number of nodes in every direction.
From Fig.~\ref{fig:two_dimensional_grid}, by connecting each first node to the last node
in each direction we obtain Fig.~\ref{fig:two_dimensional_torus}. For example,
if we connect node~$1$ with node~$4$ in the horizontal direction and
node~$1$ with node~$13$ in the vertical direction, we obtain the neighbors
of node~$1$ on the torus (nodes $2$, $4$, $5$, $13$) as shown in Fig.~\ref{fig:two_dimensional_torus}.

From the previous section,
we obtained that the spectral radius
is an important quantity to study
if our interest is the spreading
of information (or virus spreading) through the network.
The following proposition give us a relationship
between the spectral radius and the degrees of the nodes.

\begin{lemma}\cite{Lovasz2007}
Let $\mathrm{deg}_{\min}$ denote the minimum degree of~$\mathcal{G}$, let~$\overline{\mathrm{deg}}$
be the average degree of $\mathcal{G}$, and let~$\mathrm{deg}_{\max}$ be the maximum degree of~$\mathcal{G}$.
For every graph~$\mathcal{G}$,
\begin{equation*}
\max\{\overline{\mathrm{deg}},\sqrt{\mathrm{deg}_{\max}}\}\le\rho(A)\le \mathrm{deg}_{\max}.
\end{equation*}
\end{lemma}
\vspace{3mm}

We observe that for a $k$-regular graph,
its average degree and its maximum degree
are equal to~$k$, and thus $k$ corresponds to its spectral radius.
This means that for a $d$-dimensional torus,
$2d$ corresponds to its spectral radius.

The previous lemma seems to prove
that there is no interest in studying
the spectral radius over the torus
since it is a well-known quantity.
However, if the torus network
is modified (removing some of the nodes
or edges to the graph) then the propagation of
information will change (see Fig.~\ref{fig:spectral_reduction_one_node}
and Fig.~\ref{fig:spectral_reduction_two_nodes}).
In the following analysis,
we will give an explicit closed-form
expression of these changes on the
spreading.

\begin{figure*}[htbp]
  \begin{minipage}[b]{0.49\linewidth}
    \centering
    \includegraphics[width=\linewidth]{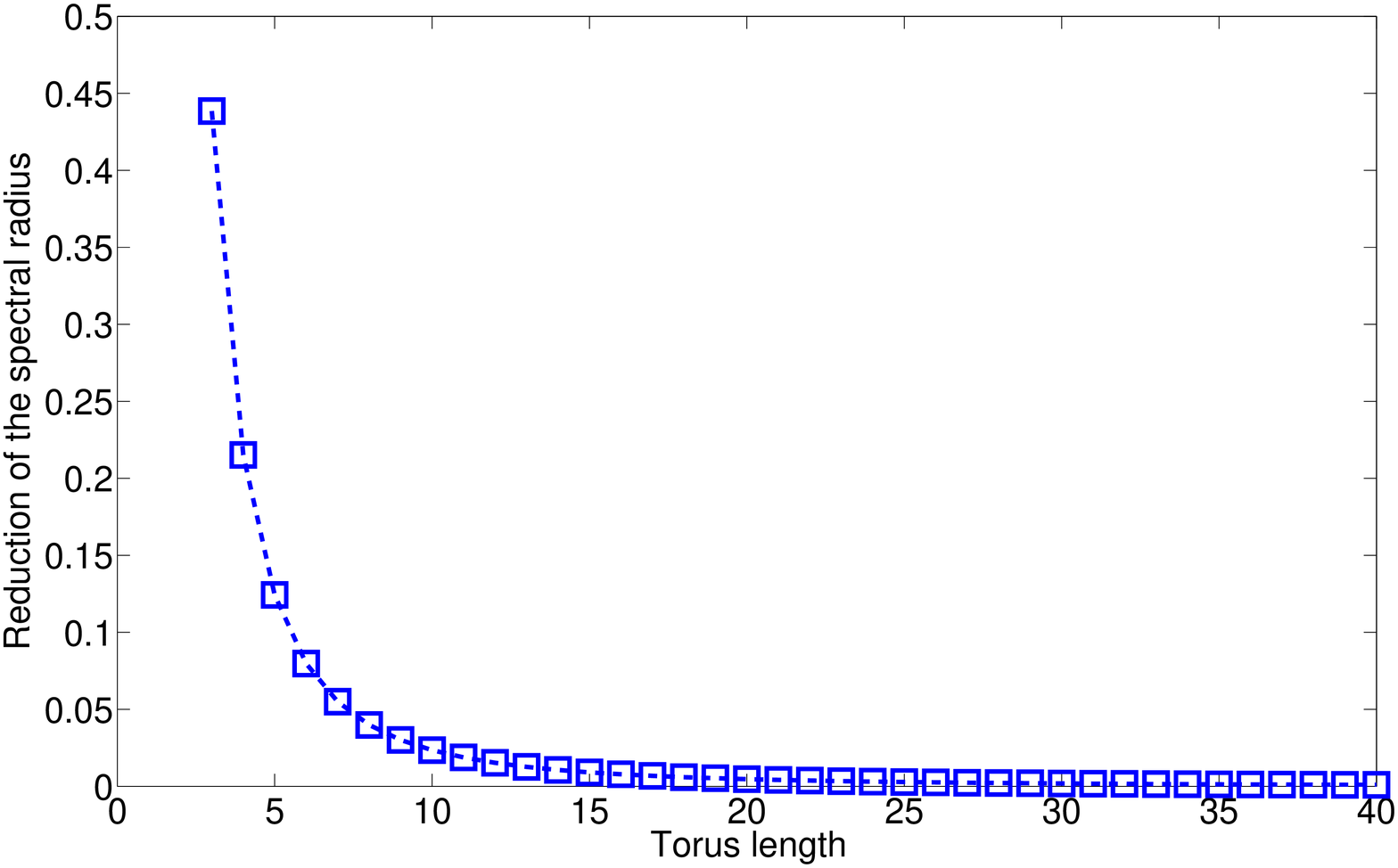}
    \caption{Spectral radius reduction by removing one node vs torus length.}
    \label{fig:spectral_reduction_one_node}
  \end{minipage}
  \hspace{0.5cm}
  \begin{minipage}[b]{0.49\linewidth}
    \centering
    \includegraphics[width=\linewidth]{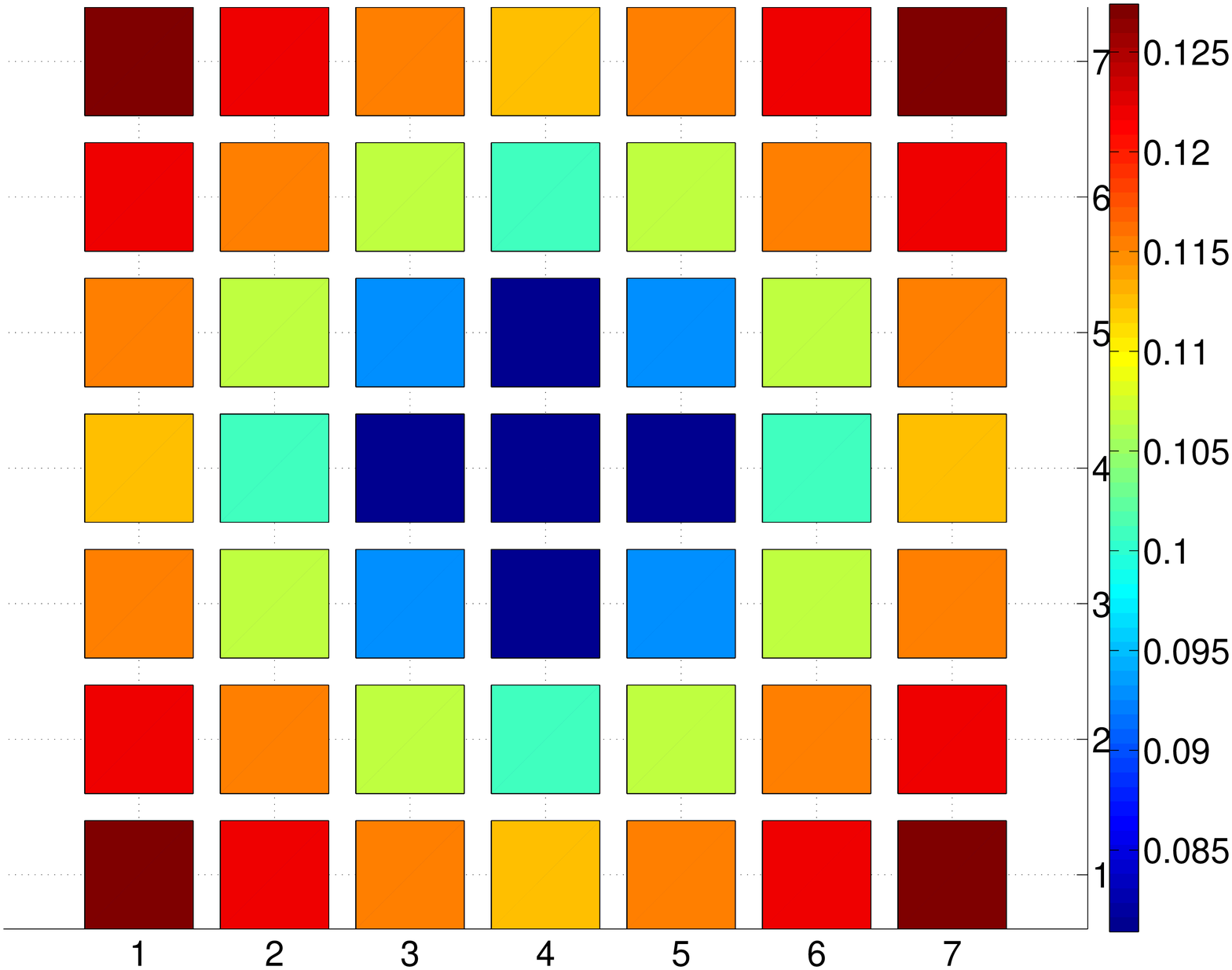}
    \caption{Spectral radius reduction by removing two nodes from a torus network of length~$7$.
One of the removed nodes is the central node.
We remove another node and we compute the spectral radius reduction and we
assign the value of this reduction at the removed node's position.}
    \label{fig:spectral_reduction_two_nodes}
  \end{minipage}
\end{figure*}

First of all, we recall a well-known result in linear algebra.
Cramer's rule~\cite{Cramer1750} states that a system of $n$ linear equations 
with $n$ unknowns, represented in matrix multiplication form $Mx=y$,
where the $n\times n$ matrix $M$ has a nonzero determinant
and the vector $x=(x_1,\ldots,x_n)^T$ is the column vector
of the unknown variables, has a unique solution,
whose individual values for the unknowns are given by
\begin{equation*}
x_i=\frac{{\mathrm{det}}(M_i)}{{\mathrm{det}}(M)},\quad\forall i\in\{1,\ldots,n\},
\end{equation*}
where $M_i$ is the matrix formed by replacing the $i$-th column
of $M$ by the column vector $y$.

We use Cramer's rule in the next lemma,
in order to establish the connection
between the characteristic polynomials
and walk generating functions.
This connection will allow us
to compute the spectral radius
of the modified torus when we remove
a node.

\begin{lemma}\cite{Godsil1992}\label{lemma:antonia}
For any graph~$\mathcal{G}$ we have
\begin{equation*}
x^{-1}W_{ii}(\mathcal{G},x^{-1})=\phi(\mathcal{G}\setminus i,x)/\phi(\mathcal{G},x).
\end{equation*}
\end{lemma}
\vspace{3mm}

\begin{proof}
We have that $W(\mathcal{G},x)=(I-xA)^{-1}=x^{-1}(x^{-1}I-A)^{-1}$.
The entries of $(x^{-1}I-A)^{-1}$ are given by Cramer's rule
(by noting that $(x^{-1}I-A)^{-1}$ corresponds to the matrix $M$
such that $(x^{-1}I-A)M=I$).
The $i$-th diagonal entry of $(x^{-1}I-A)^{-1}$ is
given by the $i$-th principal diagonal minor of $(x^{-1}I-A)$
divided by ${\mathrm{det}}(x^{-1}I-A)=\phi(\mathcal{G},x^{-1})$.
We note that the $i$-th principal diagonal minor of $(x^{-1}I-A)$ is $\phi(\mathcal{G}\setminus i, x^{-1})$,
and so the lemma follows immediately from this.
\end{proof}

\vspace{3mm}

We are interested on finding the spectral radius
of the almost regular torus, in order
to analyze the information
spreading on this topology.
In the case of the removal of one node,
this is equivalent to find
$\phi(\mathcal{G}\setminus i,x)$.
From Lemma~\ref{lemma:antonia},
we need to know the characteristic
polynomial of the regular torus $\phi(\mathcal{G},x)$
and the diagonal entries of the
walking generating function of the regular torus
$W_{ii}(\mathcal{G},x^{-1})$
since
\begin{equation}\label{eq:timbuktu23}
\phi(\mathcal{G}\setminus i,x)
=x^{-1}W_{ii}(\mathcal{G},x^{-1})\phi(\mathcal{G},x).
\end{equation}

In the next proposition,
we give an explicit expression for
the characteristic polynomial of the two-dimensional 
torus network of length~$m$.

\begin{proposition}\label{prop:characteristic_poly_torus}
The characteristic polynomial
of the two-dimensional torus
network, denoted~$\mathcal{T}^2_m$,
is given by
\begin{equation*}
\phi(\mathcal{T}^2_m,x)=
\prod_{1\le i,j\le m}\left(x-2\cos(2\pi i/m)-2\cos(2\pi j/m)\right).
\end{equation*}
\end{proposition}
\vspace{3mm}

\begin{proof}
Consider $R_m$ the ring graph which has edge set 
\mbox{$\{(u,u+1):\ 1\le u<m\}\cup\{(m,1)\}$}.
The Laplacian of $R_m$ has eigenvectors~\cite{Spielman2009}:
\begin{equation}\label{eq:eigenvectors_doumbodo}
x_k(u) = \sin(2\pi k u/m)\ \textrm{and}\ 
y_k(u) = \cos(2\pi k u/m).
\end{equation}
for~$k\le m/2$.
Both of these eigenvectors have
eigenvalue \mbox{$\lambda_k=2-2\cos(2\pi k/m)$}.
We notice that $x_0$ should be ignored and $y_0$
is the all $1$s vector. If $m$ is even, then also
$x_{m/2}$ should be ignored.

We recall that the product of two graphs $\mathcal{G}=(\mathcal{V},\mathcal{E})$ and $\mathcal{H}=(\mathcal{W},\mathcal{F})$,
denoted by $\mathcal{G}\times \mathcal{H}$, corresponds to the graph with vertex set \mbox{$\mathcal{V}\times\mathcal{W}$}
and edge set $\left((v_1,w_1),(v_2,w_2)\right)$ where \mbox{$(v_1,v_2)\in \mathcal{E}$}
and \mbox{$(w_1,w_2)\in \mathcal{F}$}.
If $\mathcal{G}$ has Laplacian eigenvalues $\lambda_1,\ldots,\lambda_m$
and eigenvectors $p_1,\ldots,p_m$;
and $\mathcal{H}$ has Laplacian eigenvalues $\mu_1,\ldots,\mu_m$
and eigenvectors $q_1,\ldots,q_m$;
then for each $1\le i\le m$ and $1\le j\le m$,
$\mathcal{G}\times \mathcal{H}$ has an eigenvector $z$ of eigenvalue $\lambda_i+\mu_j$
given by $z(v,w)=p_i(v)q_j(w)$.

In our case, the two-dimensional torus can be written as
the product of two rings, $\mathcal{G}=\mathcal{H}=R_m$,
and then $\lambda_i=\mu_i=2-2\cos(2\pi i/m)\quad\forall 1\le i\le m$.
Thus the eigenvalues of the Laplacian of the two-dimensional
torus network of length $m$ are given by
\begin{equation*}
\lambda_i+\mu_j
=4-2\cos(2\pi i/m)-2\cos(2\pi j/m)\ \forall 1\le i,j\le m.
\end{equation*}

We recall that the eigenvalues of a matrix $M$ are the solutions $\lambda$
to the equation ${\mathrm{det}}(M-\lambda I)=0$. For a $k$-regular graph of $n$ nodes we have
\begin{align*}
{\mathrm{det}}(&L-\lambda I)
={\mathrm{det}}(D-A-\lambda I)\\
&={\mathrm{det}}(-[A-(k-\lambda)I])
=(-1)^n{\mathrm{det}}(A-(k-\lambda)I),
\end{align*}
which means that $\lambda$ is an eigenvalue of $L$ if and only if $k-\lambda$ is
an eigenvalue of $A$.

In our case, the two-dimensional torus network is a $4$-regular graph,
and thus the eigenvalues of the two-dimensional
torus network of length $m$ are given by
$2\cos(2\pi i/m)+2\cos(2\pi j/m)$
for all $1\le i,j\le m$.
Thus we conclude that the characteristic polynomial is equal to
\begin{equation*}
\phi(\mathcal{T}^2_m,x)=
\prod_{1\le i,j\le m}\left(x-2\cos(2\pi i/m)-2\cos(2\pi j/m)\right).
\end{equation*}

We notice from the proof that the eigenvectors of the Laplacian coincide
with the eigenvectors of the adjacency matrix and are given by
the Kronecker product of the eigenvectors given by~\eqref{eq:eigenvectors_doumbodo}.
\end{proof}

In the following proposition,
we find an explicit expression for the diagonal 
entries of the walk generating function.

\begin{proposition}\label{prop:wgf_timbuktu}
The diagonal entries of the walk generating function of the 
two-dimensional torus of length~$m$ are:
\begin{align}\label{eq:doumbodo_321}
&W_{ii}(\mathcal{T}^2_m,x)=\\
&\sum_{\ell\ge 0}
\frac{x^{\ell}}{m^2}
\sum_{1\le i,j\le m}
4^\ell
\left(\cos(\pi(i+j)/m)\right)^\ell
\left(\cos(\pi(i-j)/m)\right)^\ell.\nonumber
\end{align}
\end{proposition}
\vspace{3mm}

\begin{proof}
We recall from Section~\ref{sec:preliminaries}
that $A^{\ell}_{ii}$ corresponds to
the closed walks of length~$\ell$.
Since each node is indistinguishable on the torus network, we know that
$
\mathrm{tr}(A^\ell)
=n\times\textrm{\# of closed walks of length~$\ell$},
$
where $n=m^2$ is the total number of nodes.
But we also have that
\begin{align*}
\mathrm{tr}(A^\ell)
&=\sum_{k=1}^{n}(\lambda_i(A))^\ell,\\
&=\sum_{1\le i,j\le m}(2\cos(2\pi i/m)+2\cos(2\pi j/m))^\ell,\\
&=\sum_{1\le i,j\le m}(4\cos(\pi (i+j)/m)\cos(\pi (i-j)/m))^\ell.
\end{align*}
Thus
\begin{equation}\label{eq:repeat}
A^{\ell}_{ii}=
\frac{1}{m^2}
\sum_{1\le i,j\le m}(4\cos(\pi (i+j)/m)\cos(\pi (i-j)/m))^\ell.
\end{equation}

Let us notice that
\begin{align}
W_{ii}(\mathcal{G},x)
={}&[(I-xA)^{-1}]_{ii}
=\left[\sum_{n\ge 0}(xA)^n\right]_{ii}
=\sum_{\ell\ge 0}x^\ell A^\ell_{ii}.\label{eq:boh345}
\end{align}
From eq.~\eqref{eq:repeat} and eq.~\eqref{eq:boh345} we conclude eq.~\eqref{eq:doumbodo_321}.
\end{proof}
\vspace{3mm}

From Proposition~\ref{prop:characteristic_poly_torus} and Proposition~\ref{prop:wgf_timbuktu}, 
we obtain the following theorem.
\begin{theorem}
The characteristic polynomial of the two-dimensional torus network
of length~$m$ where one node has been removed is given by:
\begin{align*}
&\phi(\mathcal{T}^2_m\setminus i,x)
=x^{-1}W_{ii}(\mathcal{T}^2_m,x^{-1})\phi(\mathcal{T}^2_m,x),\ \textrm{where}\\
&W_{ii}(\mathcal{T}^2_m,x^{-1})=\nonumber\\
&\sum_{\ell\ge 0}
\frac{x^{-\ell}}{m^2}
\sum_{1\le i,j\le m}
4^\ell
\left(\cos(\pi(i+j)/m)\right)^\ell
\left(\cos(\pi(i-j)/m)\right)^\ell,\\
&\phi(\mathcal{T}^2_m,x)=
\prod_{1\le i,j\le m}\left(x-2\cos(2\pi i/m)-2\cos(2\pi j/m)\right).
\end{align*}
\end{theorem}
\vspace{3mm}

We observe that all the previous calculations
do not depend on the particular removed node $i$.
This means that 
the spreading of information over
the modified torus,
if we remove one node,
is not affected by the position
of the removed node.
The importance is that one and only one node
is removed. In the next section,
we will see that this is very different
from the case of the removal of two nodes.

The previous results can also be derived for the $d$-dimensional torus network.
Since the proofs are similar to the proofs of Proposition~\ref{prop:characteristic_poly_torus} and Proposition~\ref{prop:wgf_timbuktu}
we do not include them here.

\begin{theorem}
The characteristic polynomial of the $d$-dimensional torus network
of length~$m$ where one node has been removed is given by:
\begin{equation*}
\phi(\mathcal{T}^d_m\setminus i,x)
=x^{-1}W_{ii}(\mathcal{T}^d_m,x^{-1})\phi(\mathcal{T}^d_m,x),
\end{equation*}
where the diagonal entries of the walk generating function~are
\begin{align*}
&W_{ii}(\mathcal{T}_m^d,x^{-1})=\nonumber\\
&\sum_{\ell\ge 0}
\frac{x^{-\ell}}{m^d}\!\!\!\!
\sum_{1\le i_1,i_2\ldots,i_d\le m}
\!\!\!\!\!\!\!\!\!\!\left(2\cos(2\pi i_1/m)+\ldots+2\cos(2\pi i_d/m)\right)^\ell,
\end{align*}
and the characteristic polynomial of the $d$-dimensional torus network is
given by
\begin{align*}
&\phi(\mathcal{T}_m^d,x)=\nonumber\\
&\prod_{1\le i_1,\ldots,i_d\le m} \left(x-2\cos(2\pi i_1/m)-\ldots-2\cos(2\pi i_d/m)\right).
\end{align*}
\end{theorem}

\section{Information Spreading over the Torus with a set of removed nodes}\label{sec:SetRemoved}

Following a similar approach to Lemma~\ref{lemma:antonia}, Godsil~\cite{Godsil1992}
is able to prove the following theorem.

\begin{theorem}\label{thm:nodes_removal}[Nodes removal]\cite{Godsil1992}
Let $\mathcal{S}$ be a subset of $s$ nodes from the graph~$\mathcal{G}$.
Then
\begin{equation*}
x^{-s}{\mathrm{det}}\ W_{\mathcal{S},\mathcal{S}}(\mathcal{G},x^{-1})=
\phi(\mathcal{G}\setminus \mathcal{S},x)/\phi(\mathcal{G},x),
\end{equation*}
where $W_{\mathcal{S},\mathcal{S}}(\mathcal{G},x)$ denotes the submatrix of $W(\mathcal{G},x)$
with rows and columns indexed by the elements of~$\mathcal{S}$.
\end{theorem}
\vspace{3mm}

We observe that if $\mathcal{S}$ consists of two nodes $i$ and $j$ then
\begin{equation}\label{eq:timbuktu}
{\mathrm{det}} W_{\mathcal{S},\mathcal{S}} (\mathcal{G},x) =
W_{ii}(\mathcal{G},x)W_{jj}(\mathcal{G},x)-W_{ij}(\mathcal{G},x)W_{ji}(\mathcal{G},x).
\end{equation}
and since $\mathcal{G}$ is an undirected graph, $W_{ij}(\mathcal{G},x) = W_{ji}(\mathcal{G},x)$.

For the case of the removal of two nodes, from Theorem~\ref{thm:nodes_removal} and eq.~\eqref{eq:timbuktu}, 
we obtain the following corollary.

\begin{corollary}\cite{Godsil1992}\label{cor:bamako}
For any graph~$\mathcal{G}$ we have:
\begin{align}
&x^{-1}W_{ij}(\mathcal{G},x^{-1})=\\
&\sqrt{\phi(\mathcal{G}\setminus i,x)\phi(\mathcal{G}\setminus j,x)-\phi(\mathcal{G},x)\phi(\mathcal{G}\setminus ij,x)}/\phi(\mathcal{G},x).\nonumber
\end{align}
\end{corollary}

The next corollary is extremely important since it
guarantees that independently of the number of nodes
we remove from the torus graph, we can restrict our study of the characteristic
polynomials to the case of the removal of two nodes.

\begin{corollary}\cite{Godsil1992}
If $\mathcal{C}$ is a subset of $\mathcal{V}(\mathcal{G})$ then $\phi(\mathcal{G}\setminus\mathcal{C},x)$
is determined by the polynomials $\phi(\mathcal{G}\setminus\mathcal{S},x)$
where $\mathcal{S}$ ranges over all subsets of $\mathcal{C}$ with at most two vertices.
\end{corollary}

From Corollary~\ref{cor:bamako}, we obtain that
\begin{align}
&\phi(\mathcal{G}\setminus ij,x)=\\
&\frac{\phi(\mathcal{G}\setminus i,x)\phi(\mathcal{G}\setminus j,x)}{\phi(\mathcal{G},x)}
-\phi(\mathcal{G},x)\left(x^{-1}W_{ij}(\mathcal{G},x^{-1})\right)^2.\nonumber
\end{align}

This implies that the only unknowns to compute
the removal of a set of nodes over the torus,
in particular for a set of two nodes,
are given by $W_{ij}$, which 
can be represented as the
number of walks between node~$i$
and node~$j$.
The following theorem provides us
a way to compute the values~$W_{ij}$.

\begin{theorem}\cite{Godsil1997}
For any two vertices $i$ and $j$ in the graph~$\mathcal{G}$ and any
non-negative integer~$\ell$, we have
that the number of walks of length~$\ell$,
denoted $W_{ij}(\mathcal{G},\ell)$,
is given by
\begin{equation*}
W_{ij}(\mathcal{G},\ell)=\sum_\theta u_\theta(i)^T\theta^\ell u_\theta(j),
\end{equation*}
where the sum is over all eigenvalues $\theta$ of $\mathcal{G}$
and $u_\theta(i)$ denotes the $i$-th row of $U_\theta$
where $U_\theta$ is the 
matrix
whose columns form an orthonormal basis for the eigenspace
belonging to $\theta$.
\end{theorem}
\vspace{3mm}

From the previous theorem, we conclude that
the walk generating function can be written as
\begin{equation*}
W_{ij}(\mathcal{G},x^{-1})=\sum_{\ell\ge0} x^{-\ell}\sum_\theta u_\theta(i)^T\theta^\ell u_\theta(j).
\end{equation*}

In order to compute the walk generating function of the two-dimensional torus,
we recall that the ring of length~$m$ has eigenvectors:
\begin{align*}
x_k(u) = \sin(2\pi k u/m)\quad\textrm{and}\quad y_k(u) = \cos(2\pi k u/m).
\end{align*}
for~$k\le m/2$.
Both of these eigenvectors have
eigenvalue \mbox{$\lambda_k=2-2\cos(2\pi k/m)$}.
Here $x_0$ should be ignored and $y_0$
is the all $1$s vector. If $m$ is even, then also
$x_{m/2}$ should be ignored.
We denote the matrix of eigenvectors of the ring
of length~$m$ as $V$.
Then the two-dimensional torus has matrix of eigenvectors
$Z=V\otimes V$ where $\otimes$ denotes the Kronecker product.
Since the matrix is symmetric the eigenvectors form
an orthogonal basis and we may normalize its eigenvectors to obtain
an orthonormal basis. We call this orthonormal basis~$\hat{Z}$.
We consider the matrix of eigenvalues of the two-dimensional torus,
denoted by $\Lambda$, which are given by
\begin{align*}
\lambda_i+\mu_j
={}&2\cos(2\pi i/m)+2\cos(2\pi j/m)\ \forall 1\le i,j\le m.
\end{align*}
From here we obtain the number of closed walks
of the two-dimensional torus between node~$i$ and node~$j$.

Following a similar approach to the previous case of
the removal of one node,
Godsil~\cite{Godsil1992} is able to prove
the following theorem which deals with
the removal of one edge of the graph.

\begin{theorem}[Edge removal]\cite{Godsil1992}
Let $e=\{i,j\}$ be an edge in $\mathcal{G}$. Then
\begin{align}\label{eq:edge_removal}
\phi(\mathcal{G},x)={}&
\phi(\mathcal{G}\setminus e,x)-\phi(\mathcal{G}\setminus ij,x)\\
{}&-2\sqrt{\phi(\mathcal{G}\setminus i,x)\phi(\mathcal{G}\setminus j,x)-\phi(\mathcal{G},x)\phi(\mathcal{G}\setminus ij,x)}.\nonumber
\end{align}
\end{theorem}
\vspace{3mm}

We observe that all the previous terms
in eq.~\eqref{eq:edge_removal}
except $\phi(\mathcal{G}\setminus e,x)$ are known.
This implies that the characteristic polynomial
of the two-dimensional torus with one edge removed
is completely determined by the previous expressions
for the torus with one node removed and two nodes removed.
 
\begin{figure*}[ht]
  \begin{minipage}[b]{0.49\linewidth}
    \centering
    \includegraphics[width=\linewidth]{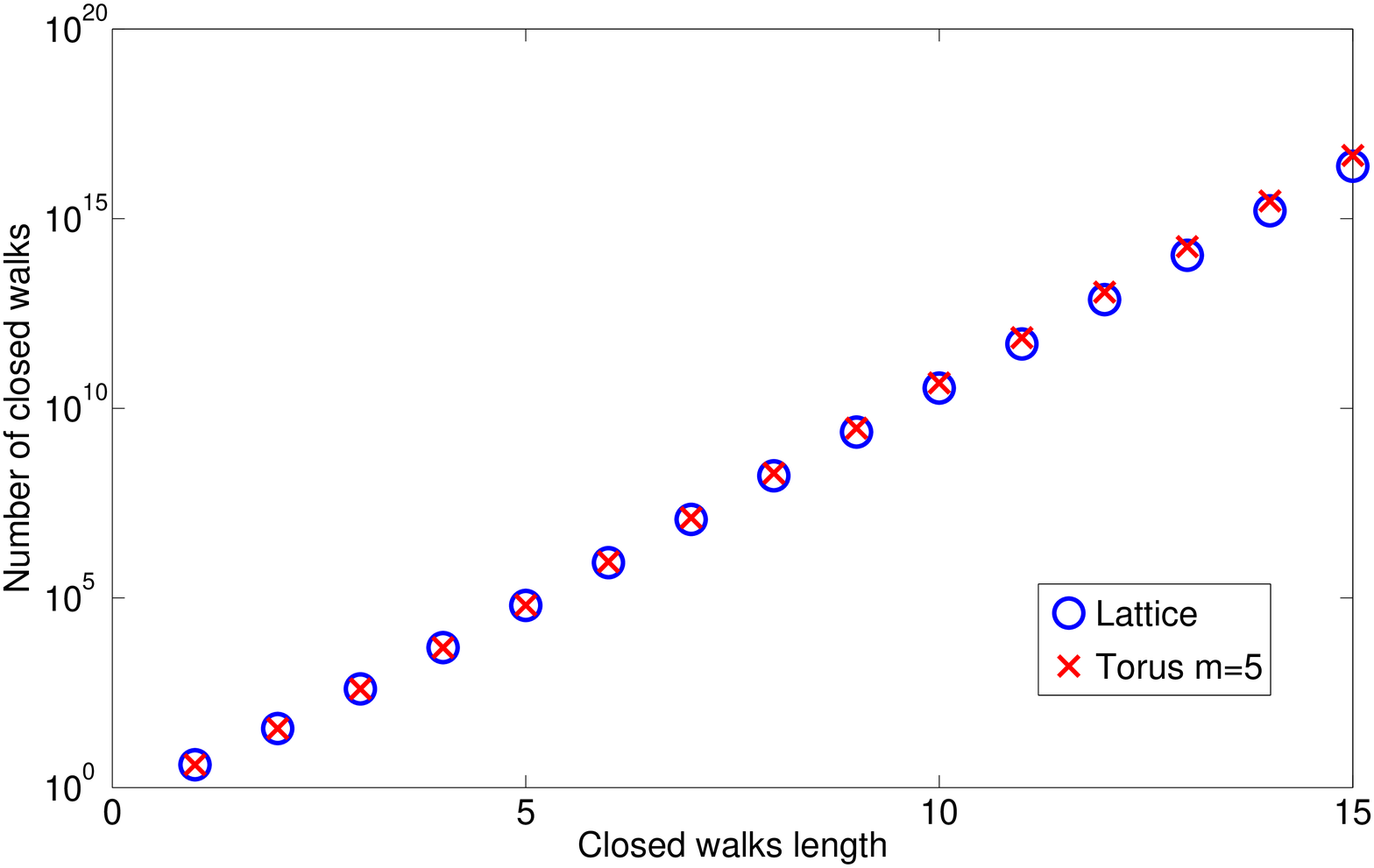}
    \caption{Number of closed walks vs closed walks length.}
    \label{fig:timbuktu1}
  \end{minipage}
  \hspace{0.5cm}
  \begin{minipage}[b]{0.49\linewidth}
    \centering
    \includegraphics[width=\linewidth]{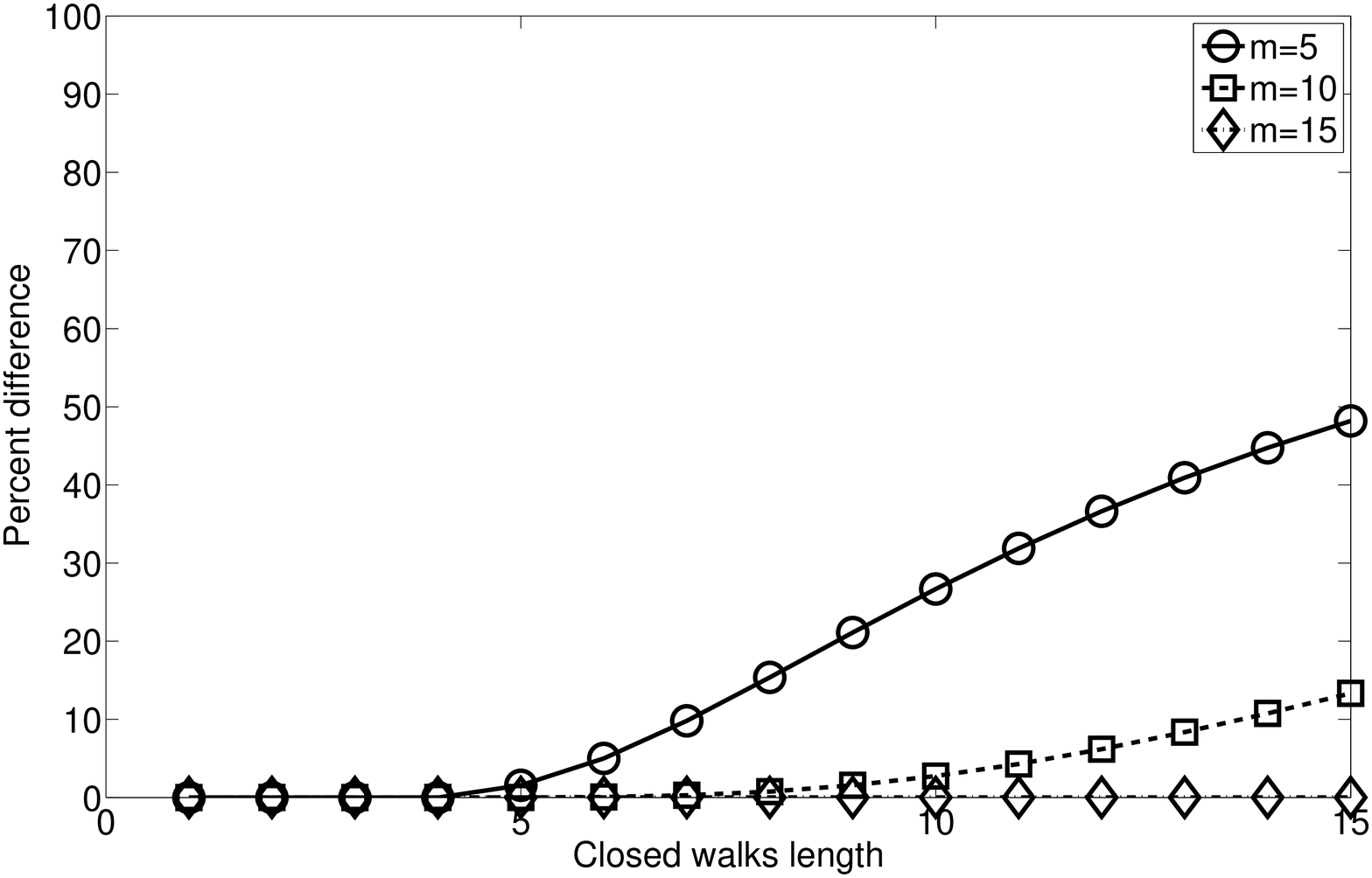}
    \caption{Percent difference on the number of closed walks between the torus and the lattice vs closed walks length.}
    \label{fig:timbuktu2}
  \end{minipage}
\end{figure*}

\begin{figure*}[ht]
  \begin{minipage}[b]{0.49\linewidth}
    \centering
    \includegraphics[width=\linewidth]{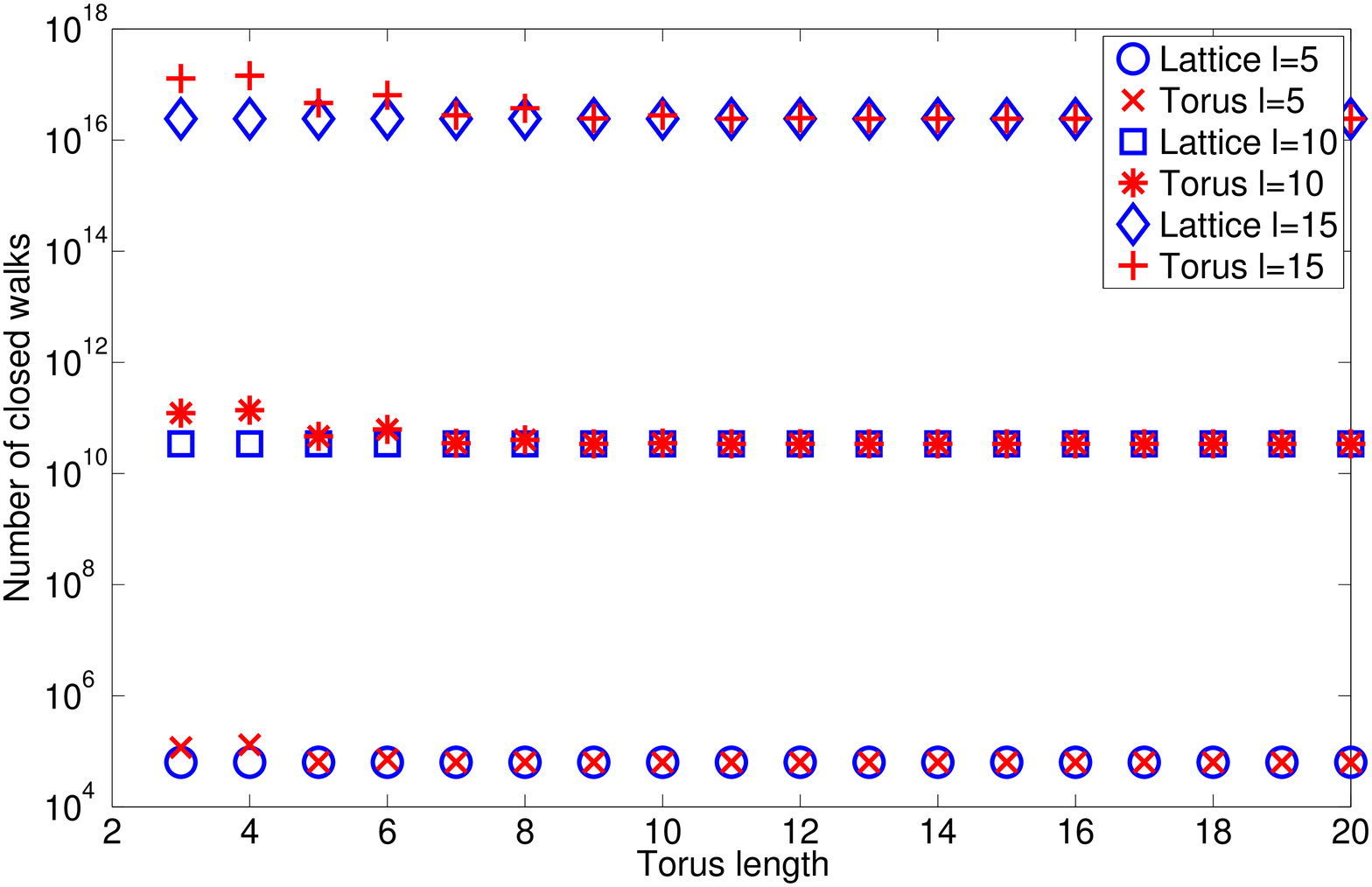}
    \caption{Number of closed walks vs torus length.}
    \label{fig:timbuktu3}
  \end{minipage}
  \hspace{0.5cm}
  \begin{minipage}[b]{0.49\linewidth}
    \centering
    \includegraphics[width=\linewidth]{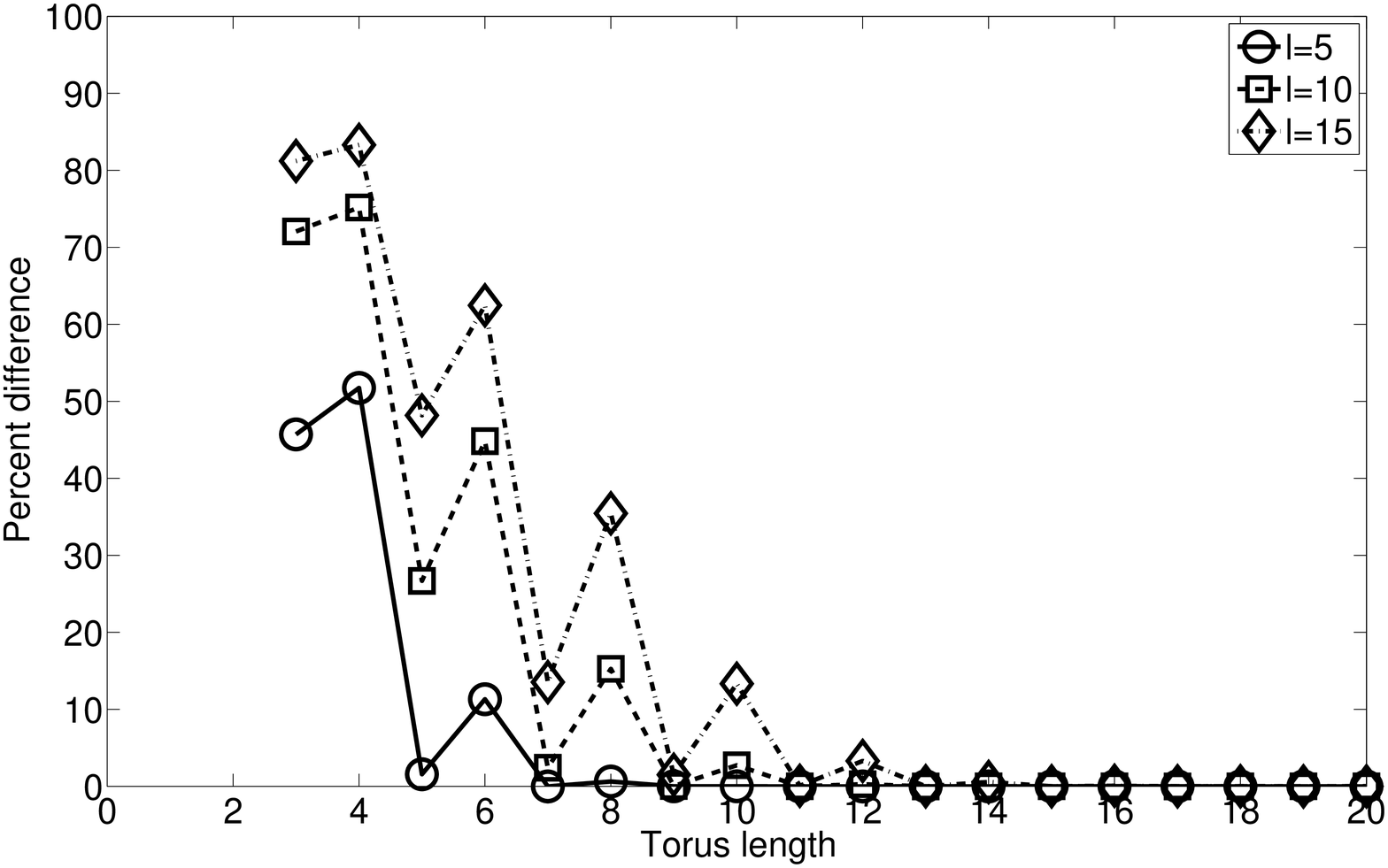}
    \caption{Percent difference on the number of closed walks between the torus and the lattice vs torus length.}
    \label{fig:timbuktu4}
  \end{minipage}
\end{figure*}

\section{Lower bounds for the two-dimensional torus networks}\label{sec:lower}

In this section, we will compare
the previous expressions with the case of the
infinite lattice. We notice that this scenario can
be seen as the case where the length of the torus goes to infinity.

\begin{proposition}\label{prop:bamako}
The diagonal entries of the walk generating function over the lattice~$\mathcal{L}^2$ are given by
\begin{equation}\label{eq:doumbodo}
W_{ii}(\mathcal{L}^2,x)
=\sum_{\ell\ge 0}x^{2\ell}\binom{2\ell}{\ell}^2.
\end{equation}
\end{proposition}

\begin{proof}
The number of closed walks of length
$2\ell$ 
over the two-dimensional lattice\footnote{
We notice that to consider even length is not restrictive since every
closed walk over the lattice network has to have even length.} $\mathcal{L}^2$ is equal to (see e.g. EIS A002894):
\begin{equation*}
\sum_{i=0}^{\ell} \frac{(2\ell)!}{i!i!\left(\ell-i\right)!\left(\ell-i\right)!}=\binom{2\ell}{\ell}^2.
\end{equation*}
\end{proof}

In the following, we give a lower bound
to the walk generating function of the two-dimensional torus
with one removed node.

\begin{proposition}\label{prop:stirling1}
For the two-dimensional torus, we have
\begin{align*}
x^{-1}W_{ii}&(\mathcal{T}_m^2,x^{-1})
\ge
\frac{1}{x}\frac{4\pi}{e^4}\sum_{\ell\ge0}\frac{1}{\ell}\left(\frac{4}{x}\right)^{2\ell}.
\end{align*}
\end{proposition}
\vspace{3mm}

\begin{proof}
The Stirling approximation formula tell us that
\begin{equation*}
\sqrt{2\pi \ell}\left(\frac{\ell}{e}\right)^\ell\le \ell! \le e\sqrt{\ell}\left(\frac{\ell}{e}\right)^\ell.
\end{equation*}

Then applying it to the binomial coefficient we get the lower bound
\begin{equation}\label{eq:doumbodo_731}
\binom{2\ell}{\ell}
=\frac{(2\ell)!}{\ell!\ell!}\ge
\frac{\sqrt{2\pi2 \ell}\left(\frac{2\ell}{e}\right)^{2\ell}}
{e\sqrt{\ell}\left(\frac{\ell}{e}\right)^{\ell e}\sqrt{\ell}\left(\frac{\ell}{e}\right)^\ell}
=\frac{\sqrt{4\pi}}{e^2}\frac{1}{\sqrt{\ell}}4^\ell.
\end{equation}

We know that the number of walks over the torus network
is greater than the number of walks over the lattice network
since every walk over the lattice network can be mapped into
a walk over the torus network. We notice that the contrary
is not true in general since for example following~$m+1$ steps in
the same direction will be a closed walk over the torus of length~$m$
but not a closed walk over the lattice.
Therefore, from eq.~\eqref{eq:doumbodo} and eq.~\eqref{eq:doumbodo_731}, we have that
\begin{align*}
x^{-1}W_{ii}&(\mathcal{T}_m^2,x^{-1})
\ge x^{-1}W_{ii}(\mathcal{L}^2,x^{-1})=\\
&= x^{-1}\sum_{\ell\ge 0}x^{-2\ell}\binom{2\ell}{\ell}^2
\ge\frac{1}{x}\frac{4\pi}{e^4}\sum_{\ell\ge0}\frac{1}{\ell}\left(\frac{4}{x}\right)^{2\ell}.
\end{align*}
\end{proof}

Proposition \ref{prop:stirling1} allow us
to obtain also a lower bound
for the characteristic polynomial
of the two-dimensional torus with one removed node.

The lower bound is given by
\begin{align*}
&\phi(\mathcal{T}^2_m\setminus i,x)
\ge
\phi(\mathcal{T}^2_m,x)\frac{1}{x}\frac{4\pi}{e^4}\sum_{\ell\ge0}\frac{1}{\ell}\left(\frac{4}{x}\right)^{2\ell}\textrm{where}\\
&\phi(\mathcal{T}^2_m,x)
=
\prod_{1\le i,j\le m}\left(x-2\cos(2\pi i/m)-2\cos(2\pi j/m)\right).
\end{align*}

The case of the number of walks between any two nodes
over the two-dimensional torus of length~$m$ is more involved.
We define the following function which
allow us to map each node on the two-dimensional torus of length~$m$
to a node on the lattice~$\mathbb{Z}\times\mathbb{Z}$.
The mapping function $h:\{1,\ldots,m^2\}\rightarrow\mathbb{Z}\times\mathbb{Z}$ is as follows:
\begin{gather*}
h(u)=(x_u,y_u)\quad\textrm{where}\\
x_u\equiv u-1\pmod{m}\quad\textrm{and}\quad
y_u=\left\lfloor \frac{u}{m}\right\rfloor.
\end{gather*}

\begin{proposition}\label{prop:thedang}
The number of walks of length $\ell$ over the two-dimensional lattice
from the origin~$(0,0)$
to~$(a,b)$ ($\ell+a+b$ is even) is equal to
\begin{equation}\label{eq:bamako231}
\sum_{i=b}^{b+(\ell-b)/2} \frac{\ell !}{i!(i-b)!\left(\frac{\ell+a+b}{2}-i\right)!\left(\frac{\ell+b-a}{2}-i\right)!}.
\end{equation}
\end{proposition}
\vspace{3mm}

\begin{proof}
We consider
$n_U$ the number of steps you go up,
$n_D$ the number of steps you go down,
$n_L$ the number of steps you go left,
$n_R$ the number of steps you go right.
The number of walks of length $\ell$ from $(0,0)$ to $(a,b)$
is equal to the number of ways you can choose $\ell$ steps
between up, down, left, and right, such that
$n_U=n_D+b$, $n_R=n_L+a$, and $n_U+n_D+n_L+n_R=\ell$.
If we fix $n_U=i$ then $n_D=i-b$, $n_L=(\frac{\ell+b-a}{2}-i)$,
$n_R=(\frac{\ell+a+b}{2}-i)$
and the number of closed walks of length $\ell$ with
$i$ steps going up is equal to
$\frac{\ell!}{i!(i-b)!(\frac{\ell+a+b}{2}-i)!(\frac{\ell+b-a}{2}-i)!}$.
Summing up for all possible values of $n_U$,
the number of paths is equal to eq.~\eqref{eq:bamako231}.
\end{proof}

From Proposition~\ref{prop:thedang} and the mapping function~$h$,
we obtain a lower bound on the number of walks of length~$\ell$
on the two-dimensional torus~$\mathcal{T}^2_m$ 
between two different nodes~$i$ and $j$
by mapping $i$ and $j$ to $h(i)=(x_i,y_i)$ and $h(j)=(x_j,y_j)$
and computing the number of walks
from $(0,0)$ to \mbox{$(x_j-x_i,y_j-y_i)$}.
We observe that the previous lower bound
is tight only when $\ell+a+b$ is even.

\section{Numerical Simulations}\label{sec:numerical}

In Fig.~\ref{fig:timbuktu1},
we compute the number of closed walks
over the torus of length $m=5$ and over the lattice vs the length of the closed walks $\ell$.
We notice that for small values of $\ell$ to approximate the number of closed walks
over the torus is more precise than for large values of $\ell$.
This fact is compensated since in the characteristic polynomial
the important terms are the terms of smaller order.
In Fig.~\ref{fig:timbuktu2},
we compute the difference on percent error of this approximation
vs the length of the closed walks $\ell$.

In Fig.~\ref{fig:timbuktu3},
we compute the number of closed walks
over the torus and over the lattice
by varying the torus length~$m$
and keeping constant the length of the closed walks~$\ell$.
We notice that for small values of the torus length $m$
the approximation is not tight, however it becomes tighter
relatively fast.
In Fig.~\ref{fig:timbuktu4},
we compute the difference on percent error of this approximation
vs the torus length $m$.

\addtolength{\textheight}{-3cm}   


\section{Conclusions}\label{sec:conclusions}
In this work, we have analyzed information spreading on almost torus networks
assuming the Susceptible-Infected-Susceptible spreading model as a metric
of information spreading.
Almost torus networks consist on the torus network topology
where some nodes or edges have been removed.
We have provided  analytical expressions for the characteristic polynomial
of these graphs and we have provided as well  tight lower bounds for its computation.
Using these expressions we are able to estimate their spectral radius 
and thus to know how the information spreads on these networks.
Simulations results that validated our analysis are presented.

\section*{Acknowledgments}

The work presented in this paper has been partially carried out at LINCS (\url{http://www.lincs.fr}).

\bibliographystyle{hieeetr}
\bibliography{mybibfile}

\begin{thebibliography}{10}

\bibitem{BrauerDJ2008}
F.~Brauer, P.~van~den Driessche, and J.~Wu, eds., {\em Mathematical
  epidemiology}, vol.~1945 of {\em Lecture Notes in Mathematics}.
\newblock Berlin: Springer-Verlag, 2008.
\newblock Mathematical Biosciences Subseries.

\bibitem{DaleyG1999}
D.~J. Daley and J.~Gani, {\em Epidemic Modeling: An Introduction}.
\newblock Cambridge, UK: Cambridge University Press, 1999.

\bibitem{DraiefM2010}
M.~Draief and L.~Massouli, {\em Epidemics and rumours in complex networks}.
\newblock Cambridge University Press, 2010.

\bibitem{KephartW1991}
J.~Kephart and S.~White, ``Directed-graph epidemiological models of computer
  viruses,'' in {\em Research in Security and Privacy, 1991. Proceedings., 1991
  IEEE Computer Society Symposium on}, pp.~343 --359, may 1991.

\bibitem{ZhangNKT2007}
X.~Zhang, G.~Neglia, J.~Kurose, and D.~Towsley, ``Performance modeling of
  epidemic routing,'' {\em Computer Networks (Amsterdam, Netherlands: 1999)},
  vol.~51, pp.~2867--2891, July 2007.

\bibitem{BampoEMSW2008}
M.~Bampo, M.~T. Ewing, D.~R. Mather, D.~Stewart, and M.~Wallace, ``The effects
  of the social structure of digital networks on viral marketing
  performance.,'' vol.~19, pp.~273--290, 2008.

\bibitem{GaneshMT2005}
A.~J. Ganesh, L.~Massouli{\'e}, and D.~F. Towsley, ``The effect of network
  topology on the spread of epidemics,'' in {\em INFOCOM}, pp.~1455--1466,
  IEEE, 2005.

\bibitem{LeskovecAH2006}
Leskovec, Adamic, and Huberman, ``The dynamics of viral marketing,'' in {\em
  CECOMM: ACM Conference on Electronic Commerce}, 2006.

\bibitem{Wei2012}
X.~Wei, N.~Valler, B.~A. Prakash, I.~Neamtiu, M.~Faloutsos, and C.~Faloutsos,
  ``Competing memes propagation on networks: a case study of composite
  networks,'' {\em Computer Communication Review}, vol.~42, no.~5, pp.~5--12,
  2012.

\bibitem{TangLDBY2011}
W.~Tang, Z.~Lan, N.~Desai, D.~Buettner, and Y.~Yu, ``Reducing fragmentation on
  torus-connected supercomputers,'' in {\em IPDPS}, pp.~828--839, IEEE, 2011.

\bibitem{TOP500}
{TOP 500 Supercomputing website}, ``http://www.top500.org,''

\bibitem{BlueGeneL}
A.~Gara, M.~A. Blumrich, D.~Chen, G.~L.-T. Chiu, P.~Coteus, M.~E. Giampapa,
  R.~A. Haring, P.~Heidelberger, D.~Hoenicke, G.~V. Kopcsay, T.~A. Liebsch,
  M.~Ohmacht, B.~D. Steinmacher-Burow, T.~Takken, and P.~Vranas, ``Overview of
  the {B}lue {G}ene/{L} system architecture,'' {\em IBM Journal of Research and
  Development}, vol.~49, pp.~195--212, march 2005.

\bibitem{Adiga2005}
N.~R. Adiga, M.~A. Blumrich, D.~Chen, P.~Coteus, A.~Gara, M.~E. Giampapa,
  P.~Heidelberger, S.~Singh, B.~D. Steinmacher-Burow, T.~Takken, M.~Tsao, and
  P.~Vranas, ``{Blue Gene/L} torus interconnection network,'' {\em IBM Journal
  of Research and Development}, vol.~49, pp.~265--276, mar./may 2005.

\bibitem{BlueGeneP}
{IBM Blue Gene team}, ``Overview of the {IBM} {B}lue {G}ene/{P} project,'' {\em
  IBM Journal of Research and Development}, vol.~52, pp.~199--220, jan. 2008.

\bibitem{Cray}
J.~Brooks and G.~Kirschner, ``Cray {XT3} and cray {XT} series of
  supercomputers,'' in {\em Encyclopedia of Parallel Computing} (D.~A. Padua,
  ed.), pp.~457--470, Springer, 2011.

\bibitem{Chen2011}
D.~Chen, N.~Eisley, P.~Heidelberger, R.~Senger, Y.~Sugawara, S.~Kumar,
  V.~Salapura, D.~Satterfield, B.~Steinmacher-Burow, and J.~Parker, ``The ibm
  blue gene/q interconnection network and message unit,'' in {\em High
  Performance Computing, Networking, Storage and Analysis (SC), 2011
  International Conference for}, pp.~1 --10, nov. 2011.

\bibitem{Ajima2009}
Y.~Ajima, S.~Sumimoto, and T.~Shimizu, ``Tofu: {A} 6{D} mesh/torus interconnect
  for exascale computers,'' {\em IEEE Computer}, vol.~42, pp.~36--40, Nov.
  2009.
\newblock Fujitsu.

\bibitem{Hethcote2000}
H.~W. Hethcote, ``The mathematics of infectious diseases,'' {\em SIAM Review},
  vol.~42, pp.~599--653, Dec. 2000.

\bibitem{Wang2003}
Y.~Wang, D.~Chakrabarti, C.~Wang, and C.~Faloutsos, ``Epidemic spreading in
  real networks: An eigenvalue viewpoint,'' in {\em SRDS}, pp.~25--34, IEEE
  Computer Society, 2003.

\bibitem{Chakrabarti2008}
D.~Chakrabarti, Y.~Wang, C.~Wang, J.~Leskovec, and C.~Faloutsos, ``Epidemic
  thresholds in real networks,'' {\em ACM Trans. Inf. Syst. Secur}, vol.~10,
  no.~4, 2008.

\bibitem{Hirsch1974}
M.~W. Hirsch and S.~Smale, {\em Differential equations, dynamical systems, and
  linear algebra}.
\newblock New York: Academic Press, 1974.

\bibitem{Lovasz2007}
L.~Lov{\'a}sz, ``Eigenvalues of graphs,''

\bibitem{Cramer1750}
G.~Cramer, ``Introduction {\`a} l'analyse des lignes courbes alg{\'e}briques
  (in french),'' {\em Geneva: Europeana}, pp.~656--659, 1750.

\bibitem{Godsil1992}
C.~D. Godsil, ``Walk generating functions, christoffel-darboux identities and
  the adjacency matrix of a graph,'' {\em Combinatorics, Probability and
  Computing}, vol.~1, no.~01, pp.~13--25, 1992.

\bibitem{Spielman2009}
D.~Spielman, ``Spectral graph theory and its applications,''

\bibitem{Godsil1997}
A.~Chan and C.~D. Godsil, ``Symmetry and eigenvectors,'' in {\em Graph symmetry
  ({M}ontreal, {PQ}, 1996)}, vol.~497 of {\em NATO Adv. Sci. Inst. Ser. C Math.
  Phys. Sci.}, pp.~75--106, Dordrecht: Kluwer Acad. Publ., 1997.

\end{thebibliography}

\end{document}